\documentclass[10pt,conference]{IEEEtran}\IEEEoverridecommandlockouts
\usepackage{amsmath,amsfonts}
\usepackage{array}
\usepackage{textcomp}
\usepackage{stfloats}
\usepackage{url}
\usepackage{verbatim}
\usepackage{graphicx}
\usepackage{subcaption}
\usepackage{cite}
\usepackage{verbatim}
\usepackage{booktabs}
\usepackage{algorithm}  
\usepackage{algpseudocode}  
\usepackage{amssymb}
\usepackage{bm}
\usepackage{mathrsfs}
\usepackage{float}
\usepackage{subfig}
\usepackage{overpic}
\usepackage{soul}
\usepackage{color}
\usepackage{amsthm}

\usepackage{siunitx}
\usepackage{caption}    
\usepackage{breqn}

\newtheorem{theorem}{Theorem}[section]

\usepackage[letterpaper, top=0.7in, bottom=0.99in, left=0.65in, right=0.65in]{geometry}

\usepackage[numbers,compress]{natbib}

\usepackage{braket}

\begin{document}

    \title{Channel-Robust RFF for Low-Latency 5G Device Identification in SIMO Scenarios}
    
\author{
    \IEEEauthorblockN{
    Yingjie Sun$^1$, Guyue Li$^{1,2}$, Hongfu Chou$^3$ and Aiqun Hu$^{2,4}$
    }
    
    \IEEEauthorblockA{
    $1$ School of Cyber Science and Engineering, Southeast University, Nanjing, China
    }
    \IEEEauthorblockA{
    $2$ Purple Mountain Laboratories, Nanjing, China
    }
    \IEEEauthorblockA{
    $3$ Interdisciplinary Centre for Security, Reliability and Trust (SnT), University of Luxembourg, Luxembourg
    }
    \IEEEauthorblockA{
    $4$ National Mobile Communications Research Laboratory, Southeast University, Nanjing, China
    }
    \IEEEauthorblockA{
    Email: \{sun\_yingjie, guyuelee\}@seu.edu.cn, hungpu.chou@uni.lu, aqhu@seu.edu.cn
    }
}



\maketitle
\pagestyle{empty}  
\thispagestyle{empty} 

\begin{abstract}
Ultra-low latency, the hallmark of fifth-generation mobile communications (5G), imposes exacting timing demands on identification as well. 
Current cryptographic solutions introduce additional computational overhead, which results in heightened identification delays. Radio frequency fingerprint (RFF) identifies devices at the physical layer, blocking impersonation attacks while significantly reducing latency. Unfortunately, multipath channels compromise RFF accuracy, and existing channel-resilient methods demand feedback or processing across multiple time points, incurring extra signaling latency. 
To address this problem, the paper introduces a new RFF extraction technique that employs signals from multiple receiving antennas to address multipath issues without adding latency. Unlike single-domain methods, the Log-Linear Delta Ratio (LLDR) of co-temporal channel frequency responses (CFRs) from multiple antennas is employed to preserve discriminative RFF features, eliminating multi-time sampling and reducing acquisition time. To overcome the challenge of the reliance on minimal channel variation, the frequency band is segmented into sub-bands, and the LLDR is computed within each sub-band individually. 
Simulation results indicate that the proposed scheme attains a 96.13\% identification accuracy for 30 user equipments (UEs) within a 20-path channel under a signal-to-noise ratio (SNR) of 20 dB. Furthermore, we evaluate the theoretical latency using the Roofline model, resulting in the air interface latency of 0.491~ms, which satisfies ultra-reliable and low-latency communications (URLLC) latency requirements.

\end{abstract}


\section{Introduction}

The explosive growth of wireless devices, coupled with the open nature of the radio channel
, has rendered conventional cryptographic identification unwieldy under stringent key-management and latency constraints \cite{PATWARY2020212}. This drives demand for a low-latency, tamper-proof, and highly secure scheme. As a remedy, radio frequency fingerprint (RFF) exploits device-unique distortions at the physical layer \cite{HD2020} introduced by manufacturing tolerances and environmental conditions to enable low-latency, pre-transmission identification \cite{10978824}. However, multipath propagation compromises feature fidelity, challenging the reliability required for accurate identification \cite{9155259}.

Current research on alleviating the impact of multipath channels on RFF is predominantly concentrated on the following two aspects:
\begin{itemize}
    \item \textbf{Channel elimination based on feedback:}
    The DeepCSI method in \cite{9912203} employs deep learning of beamforming matrices to extract device-specific features from channel state information (CSI) feedback, leveraging imperfections in the transmitter's RF circuitry embedded in plaintext beamforming matrices.
    \item \textbf{Channel elimination based on adjacent symbol spectrum:}
    The DoLoS method in \cite{9979789} extracts channel-independent RFF features by leveraging the logarithmic spectral differences between adjacent signal symbols transmitted within the channel coherence time, during which the channel is assumed constant.
\end{itemize}

However, feedback-based channel elimination requires multiple interactions that increase identification latency, making it unsuitable for large-scale device identification scenarios. Moreover, practical RFF differences between consecutive symbols remain negligible; thus, the logarithmic spectral differencing used in this approach risks significant RFF feature degradation.



In the realm of low-latency RFF, the open-set RFF in \cite{electronics13020384} prunes the exemplar set to the most discriminative features, cutting identification latency by 58\% successfully yet marginally degrading accuracy, and omits verification of whether the resulting air interface delay meets ultra-reliable and low-latency communications (URLLC) requirements.

Inspired by \cite{10965708}, which solely exploits signals acquired by a single antenna across disparate temporal instances, we extend our focus to multi-antenna systems, leveraging the received signals of different antennas to eliminate the impact of multipath channel on RFF. This approach enables suppression of multipath effects using only a single uplink transmission, while preserving more transmitter-specific RFF features.


In summary, we present a novel RFF extraction scheme for single-input multiple-output (SIMO) systems that leverages multi-antenna signals. Our main contributions are summarized as follows.

\begin{itemize}
    \item We propose a novel RFF scheme that utilizes the received signals from multiple antennas to mitigate the impact of the multipath channel, requiring only a single uplink transmission and thereby reducing identification latency. It can be theoretically demonstrated that this RFF method is robust to multipath channels.
    \item Unlike traditional signal processing in a single domain, we calculate the Log-Linear Delta Ratio (LLDR) between the estimated channel frequency responses (CFRs). To address CFR variability, we partition the frequency band and extract RFF features within each smooth sub-band.
    \item We simulate different channel conditions using the MATLAB 5G NR Toolbox for 30 virtual user equipments (UEs) and different signal-to-noise ratio (SNR) levels. Simulation results show that the proposed scheme can achieve 96.13\% accuracy under the tapped delay line (TDL) 20-path channel with an SNR of 20~dB and 0.491~ms air interface latency, satisfying the URLLC latency requirements.
\end{itemize}

\section{System Model and Problem Statement}\label{Sec:II}

In this section, we construct the system model and signal model of the proposed RFF method, formulate the problem statement and analyze the challenges faced by RFF extraction.

\subsection{System Model}

We propose a generalized communication model for SIMO systems, as shown in Fig.~\ref{fig:system_model}. The system consists of a single-antenna UE and a multi-antenna base station (BS). The UE sends a pilot signal to the BS for identification. Subsequent device identification is performed at the BS by utilizing the received signals at different antennas.

\begin{figure}[h]
    \centering  
    \includegraphics[width=0.8\linewidth]{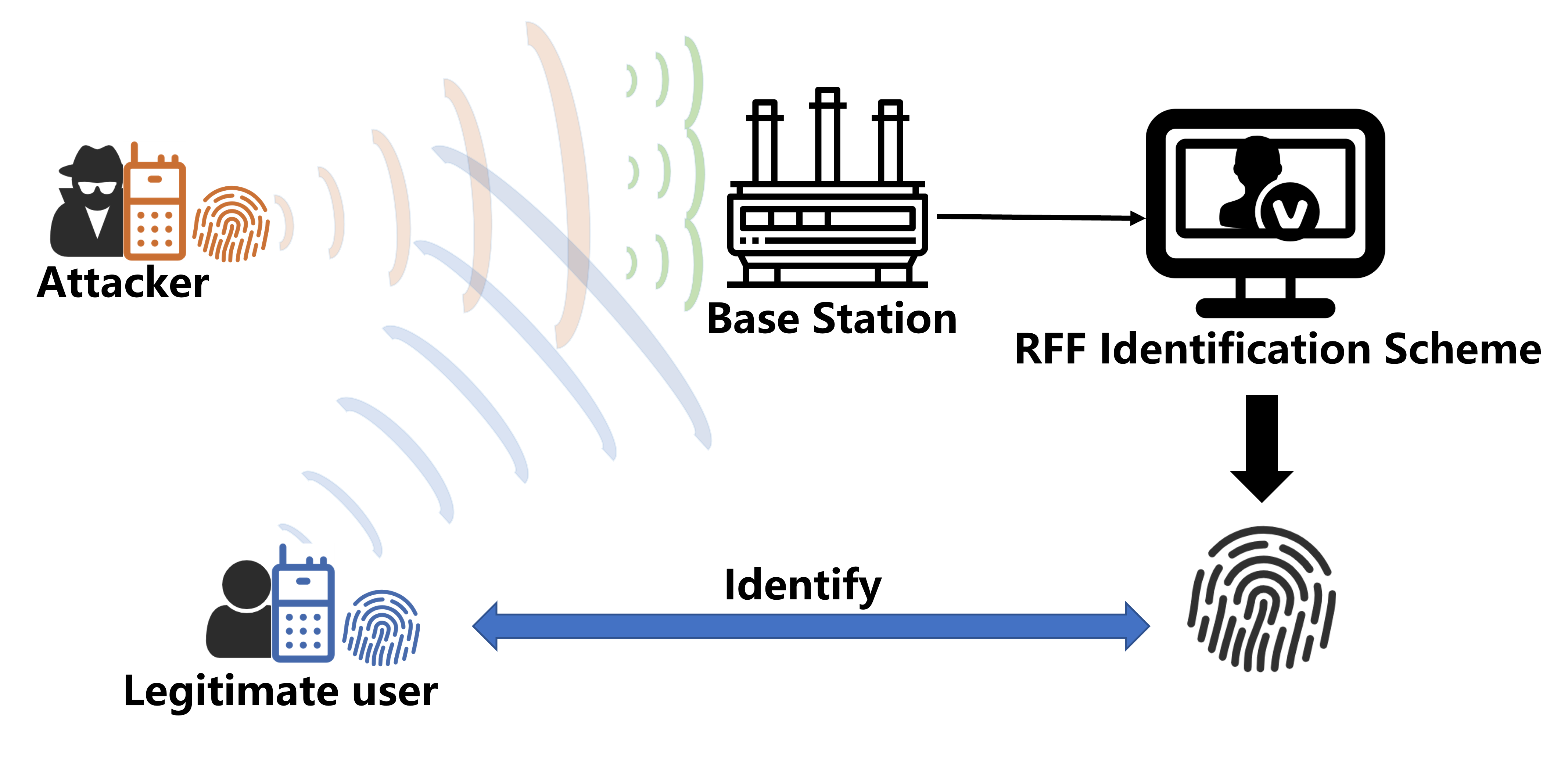}
    \caption{System model.}
    \label{fig:system_model}
\end{figure}


To streamline exposition, the algorithm is derived under the baseline scenario of a two-antenna BS, wherein the RFF is recovered exclusively from the signals captured at receive antennas 1 and 2. In the frequency domain, the received signals of the two receive antennas are respectively given by
\begin{align}
    Y_1(k) &= R_1(k) \cdot H_1(k) \cdot T(k) \cdot X(k) + N_1(k), \label{eq:1} \\
    Y_2(k) &= R_2(k) \cdot H_2(k) \cdot T(k) \cdot X(k) + N_2(k), \label{eq:2}
\end{align}
where $k$ denotes the index of the subcarrier, $X(k)$ represents the pilot signal transmitted by UE, $Y_1 (k)$ and $Y_2 (k)$ are the received signals on receive antennas 1 and 2, respectively. $T(k)$ denotes the transmission defect of the UE, while $R_1 (k)$ and $R_2 (k)$ represent the reception defects of receive antennas 1 and 2, respectively. $H_1 (k)$ and $H_2 (k)$ are the CFRs from the UE to receive antennas 1 and 2, respectively. $N_1 (k)$ and $N_2 (k)$ are modeled as independent and identically distributed additive white Gaussian noise.

Considering that communication is conducted in a 5G system, we refer to the RFF modeling in \cite{10279540} incorporating the following hardware-specific characteristics:
\begin{enumerate}
    \item \textbf{In-phase/quadrature (IQ) imbalance:} 
    Each RF chain contains independent IQ mismatches with gain imbalance $\alpha \in [-1,1]$ dB and phase offset $\beta \in [-5,5]^\circ$ \cite{9450821};
    \item \textbf{Oscillator:}
    All antennas of the same device share a common oscillator, resulting in identical carrier frequency offset (CFO) and common phase offset (CPO);
    \item \textbf{Power amplifier (PA) nonlinearity:}
    A generalized memory polynomial model with memory depth 5 and nonlinear order 5 is adopted, whose coefficients $\{a_{mk}\}$ are perturbed within $\pm5\%$ of measured values \cite{1703853}.
\end{enumerate}
The parameter set $\Omega = \{\alpha, \beta, \text{CFO}, \text{CPO}, \{a_{mk}\}\}$ uniquely defines device-specific RFF for $T(k)$, $R_1 (k)$ and $R_2 (k)$.

The channel between BS and UE is modeled as a TDL multipath channel \cite{3GPPTR38901}, with its frequency-domain representation given by
\begin{equation}\label{eq:3}
    H(k) = \sum_{i=0}^{L-1} a_i e^{-j2\pi k \Delta f \tau_i},
\end{equation}
where \( L \) denotes the number of paths, \( a_i \) and \( \tau_i \) represent the path gain and delay of the \( i \)-th path, respectively, and \( \Delta f \) is the subcarrier spacing.
Consequently, $H_1 (k)$ and $H_2 (k)$ are likewise modeled as independent TDL channels.

On the basis of the aforementioned modeling, the channel estimation results obtained based on the received signals of the two receive antennas are respectively given by
\begin{align}
Z_1 (k) &\approx Y_1 (k) \cdot X^{-1} (k) \approx R_1 (k) \cdot H_1 (k) \cdot T(k), \label{eq:4} \\
Z_2 (k) &\approx Y_2 (k) \cdot X^{-1} (k) \approx R_2 (k) \cdot H_2 (k) \cdot T(k), \label{eq:5}
\end{align}

\subsection{Problem Statement}



Our objective is to extract channel-robust RFF from the received signal. As indicated by \eqref{eq:4} and \eqref{eq:5}, the channel estimation results incorporate the transmitter's RFF \( T(k) \). Although \( T(k) \) is theoretically extractable, the superposition of \( H_1(k) \) and \( H_2(k) \) onto \( T(k) \) fundamentally precludes direct acquisition of stable RFF features required for reliable identification.

Consequently, the goal of our proposed method is to mitigate the impact of \( H_1(k) \) and \( H_2(k) \) to obtain a stable RFF feature only through a single uplink transmission.

\section{Channel-Robust and Low-Latency RFF Extraction}\label{Sec:III}


In this section, we propose a multi-antenna channel-robust RFF scheme that utilizes the received signals from multiple antennas, delineating the feature extraction procedure and conducting a theoretical analysis of the method's performance.

\subsection{Preliminary}


The proposed RFF extraction method partitions the frequency band to ensure CFR smoothness and computes the LLDR of CFRs within each sub-band while compensating for the disparity between receive antennas. The subsequent section formalizes this process after presenting preliminary knowledge.


\subsubsection{Frequency Band Segmentation}


To address the significant fluctuations of CFR, the entire frequency band is divided into smaller sub-bands across subcarriers. Specifically, the CFR within each sub-band can be represented as
\begin{align}
H_{1}(k) &= \bar{H}_{1}(k) + \Delta H_{1}(k), \label{eq:6} \\
H_{2}(k) &= \bar{H}_{2}(k) + \Delta H_{2}(k), \label{eq:7}
\end{align}
where $\bar{H}_{1}(k)$ and $\bar{H}_{2}(k)$ are the mean values of the CFR within the sub-bands, and $\Delta H_{1}(k)$ and $\Delta H_{2}(k)$ are vanishingly small values that characterize the deviation from its mean.


\subsubsection{Approximation Based on Taylor Expansion}


A value $x$ fluctuating around the constant value of 1 can be represented as $x = 1 + \Delta x$, where $\Delta x$ is a vanishingly small value approximated as $\Delta x \to 0$. Based on the first-order Taylor expansion, the natural logarithm of $x$ can be expanded as
\begin{equation}\label{eq:8}
    \ln x = \ln(1 + \Delta x) = \Delta x + e(\Delta x).
\end{equation}
where
$e(\Delta x) = -\frac{1}{2} (\Delta x)^2 + \cdots + \frac{(-1)^{n-1}}{n} (\Delta x)^n + O((\Delta x)^{n+1}).$
When $\Delta x \to 0$, $e(\Delta x)$ can be considered as a higher-order infinitesimal of $\Delta x$ and thus can be neglected. In other words, $\ln x$ can be approximated as
\begin{equation}\label{eq:9}
    \ln x \approx \Delta x.
\end{equation}

\subsubsection{Compensation for the Disparity between Receive Antennas}


Let the relative impairment between receive antenna 1 and receive antenna 2 be defined as
$r = \frac{R_{2}(k)}{R_{1}(k)}.$
Then, the compensated linear and logarithmic differences can be expressed as
\begin{align}
D_{12}(k) &= Z_{1}(k) - \frac{1}{r} Z_{2}(k) \notag \\
&= R_{1}(k) \cdot (H_{1}(k) - H_{2}(k)) \cdot T(k), \label{eq:10} \\
L_{12}(k) &= \ln r + \ln \frac{Z_{1}(k)}{Z_{2}(k)} \notag \\
&= \ln H_{1}(k) - \ln H_{2}(k). \label{eq:11}
\end{align}
At this point, the results have been compensated for receive antenna impairments and can be employed to extract RFF.

\subsection{RFF Extraction Procedure}

Algorithm~\ref{alg:ext} outlines the RFF extraction process: BS conducts dual-antenna channel estimation, calculates sub-band means as CFR estimates, and normalizes channel estimations accordingly. Then LLDR is computed with inter-antenna imperfection compensation.

\begin{algorithm}
\caption{The Extraction of UE's RFF}
\label{alg:ext}
\begin{algorithmic}[1]
    \Require  
        \parbox[t]{\dimexpr\linewidth-\algorithmicindent-\algorithmicindent-\algorithmicindent}{\( X(k) \), \( Y_1(k) \) , \( Y_2(k) \), and the relative defect of the receive antennas \( r \)}
    \Ensure  
        The RFF of UE, \( \hat{T}_0(k) \)
    \State BS estimates channels \( Z_1(k) \approx Y_1(k) \cdot X^{-1}(k) \) and \( Z_2(k) \approx Y_2(k) \cdot X^{-1}(k) \), respectively;

    \For{each sub-band \( (k_i, k_{i+1}) \)}
        \State \parbox[t]{\dimexpr\linewidth-\algorithmicindent-\algorithmicindent-\algorithmicindent}{The mean values of the channel estimation results are calculated as \( \hat{a}_1 \) and \( \hat{a}_2 \);}
        \State \( Z_{1,\text{norm}}(k) = \frac{Z_1(k)}{\hat{a}_1} \); 
        \State \( Z_{2,\text{norm}}(k) = \frac{Z_2(k)}{\hat{a}_2} \);
        \State \( D_{12}(k) = Z_{1,\text{norm}}(k) - \frac{1}{r} Z_{2,\text{norm}}(k) \);
        \State \( L_{12}(k) = \ln r + \ln \frac{Z_{1,\text{norm}}(k)}{Z_{2,\text{norm}}(k)} \);
        \State \( \hat{T}_0(k) = \frac{D_{12}(k)}{L_{12}(k)} \);
    \EndFor
    \State The resulting RFF estimates \( \hat{T}_0(k) \) are concatenated to complete the UE's RFF;
    \State \textbf{return} \( \hat{T}_0(k) \);
\end{algorithmic}
\end{algorithm}


In the following, the focus will be on the computation within a specific sub-band, denoted as $(k_i, k_{i+1})$.


The mean values of the CFRs within the sub-band are considered as the mean channel estimations \( \hat{a}_1 \) and \( \hat{a}_2 \). Analogous to \eqref{eq:6} and \eqref{eq:7}, the CFRs can be expressed as
\begin{align}
H_1(k) &= \hat{a}_1 + \Delta H_1(k), \label{eq:12} \\
H_2(k) &= \hat{a}_2 + \Delta H_2(k), \label{eq:13}
\end{align}
where \( \Delta H_1(k) \) and \( \Delta H_2(k) \) are vanishingly small values.

By normalizing \eqref{eq:4} and \eqref{eq:5} using \( \hat{a}_1 \) and \( \hat{a}_2 \), respectively, we obtain
\begin{align}
Z_{1,\text{norm}}(k) &= \frac{Z_1(k)}{\hat{a}_1} \approx R_1(k) \cdot \frac{H_1(k)}{\hat{a}_1} \cdot T(k), \label{eq:14} \\
Z_{2,\text{norm}}(k) &= \frac{Z_2(k)}{\hat{a}_2} \approx R_2(k) \cdot \frac{H_2(k)}{\hat{a}_2} \cdot T(k), \label{eq:15}
\end{align}

According to \eqref{eq:10} and \eqref{eq:11}, to compensate for inter-antenna imbalance, the linear and logarithmic differences of \( Z_{1,\text{norm}}(k) \) and \( Z_{2,\text{norm}}(k) \) are constructed as follows:
\begin{align}
D_{12}(k) &= Z_{1,\text{norm}}(k) - \frac{1}{r} Z_{2,\text{norm}}(k), \label{eq:16} \\
L_{12}(k) &= \ln r + \ln \frac{Z_{1,\text{norm}}(k)}{Z_{2,\text{norm}}(k)}. \label{eq:17}
\end{align}

We define \( T_0(k) = R_1(k) \cdot T(k) \). By taking the ratio of \eqref{eq:16} and \eqref{eq:17}, the RFF estimate of UE within the sub-band $(k_i, k_{i+1})$ is given by
\begin{equation}\label{eq:18}
\hat{T}_0(k) = \frac{D_{12}(k)}{L_{12}(k)}.
\end{equation}

\begin{theorem}
    The RFF estimate $\hat{T}_0(k)$ is approximately equal to the true value of RFF $T_0(k)$, that is,
    \begin{equation}\label{eq:19}
        \hat{T}_0(k) \approx T_0(k).
    \end{equation}
\end{theorem}

\begin{proof}
    The normalized channels can be approximated as
    \begin{align}
        H_{1,\text{norm}}(k) &= \frac{H_1(k)}{\hat{a}_1} = 1 + \Delta H_1'(k), \label{eq:20} \\
        H_{2,\text{norm}}(k) &= \frac{H_2(k)}{\hat{a}_2} = 1 + \Delta H_2'(k). \label{eq:21}
    \end{align}
    where \( \Delta H_1'(k) \) and \( \Delta H_2'(k) \) are also vanishingly small values.

    Thus, \eqref{eq:16} and \eqref{eq:17} can be written as
    \begin{align}
        D_{12}(k) &= T_0(k) \cdot \left( \Delta H_1'(k) - \Delta H_2'(k) \right), \label{eq:24} \\
        L_{12}(k) &= \ln \left(1 + \Delta H_1'(k)\right) - \ln \left(1 + \Delta H_2'(k)\right). \label{eq:25}
    \end{align}
    
    Given that \( \Delta H_1'(k) \) and \( \Delta H_2'(k) \) are vanishingly small, according to \eqref{eq:9}, the logarithmic difference of \( Z_{1,\text{norm}}(k) \) and \( Z_{2,\text{norm}}(k) \) can be approximated as
    \begin{equation}\label{eq:28}
        L_{12}(k) \approx \Delta H_1'(k) - \Delta H_2'(k).
    \end{equation}
    
    Hence, by taking the ratio of \eqref{eq:24} and \eqref{eq:28}, we can prove that
    \begin{equation}\label{eq:29}
        \hat{T}_0(k) = \frac{D_{12}(k)}{L_{12}(k)} \approx T_0(k).
    \end{equation}
\end{proof}


Repeating the procedure for each sub-band and concatenating the estimates \(\hat{T}_0(k)\) yields a channel-robust RFF, since \(\hat{T}_0(k)\) is observed independent of \(H_1(k)\) and \(H_2(k)\) from \eqref{eq:29}, indicating the successful mitigation of the multipath channel.




\subsection{Impact of Channel and Noise on RFF Estimation Accuracy}


In the algorithm description, the results are considered under theoretical assumptions. However, in practical wireless environments, the influence of additive noise becomes significant and cannot be neglected. 

The TDL channels \( H_1(k) \) and \( H_2(k) \) are modeled as independent and identically distributed complex Gaussian variables:
$
H_1(k) \sim \mathcal{CN}(0, \sigma_H^2), \quad H_2(k) \sim \mathcal{CN}(0, \sigma_H^2),
$
and the Gaussian noises \( N_1(k) \) and \( N_2(k) \) are similarly modeled:
$
N_1(k) \sim \mathcal{CN}(0, \sigma_N^2), \quad N_2(k) \sim \mathcal{CN}(0, \sigma_N^2).
$

By recalculating the RFF estimate and estimating the variance, we obtain
\begin{equation}\label{eq:31}
    \text{Var}(\hat{T}_0(k)) \approx c_1 \sigma_H^8 + c_2 \sigma_H^6 \sigma_N^2 + c_3 \sigma_H^4 \sigma_N^4 + c_4 \sigma_H^2 \sigma_N^6,
\end{equation}
where \( c_1 \), \( c_2 \), \( c_3 \), and \( c_4 \) are constants that do not depend on the channel or noise realizations. 
As shown in \eqref{eq:31}, the variance of the RFF estimate is affected to some extent by the propagation channel conditions. Nevertheless, our subsequent experimental results demonstrate that, even under such influence, the proposed scheme is still able to achieve a satisfying recognition accuracy.

\section{Proposed RFF Identification Framework in SIMO Systems}

In this section, we propose an RFF identification framework that leverages the RFF extraction method in Section \ref{Sec:III}. As shown in Fig.~\ref{fig:system_overview}, our proposed method consists of the training stage and the inference stage.

\begin{figure}[h]
    \centering  
    \includegraphics[width=0.8\linewidth]{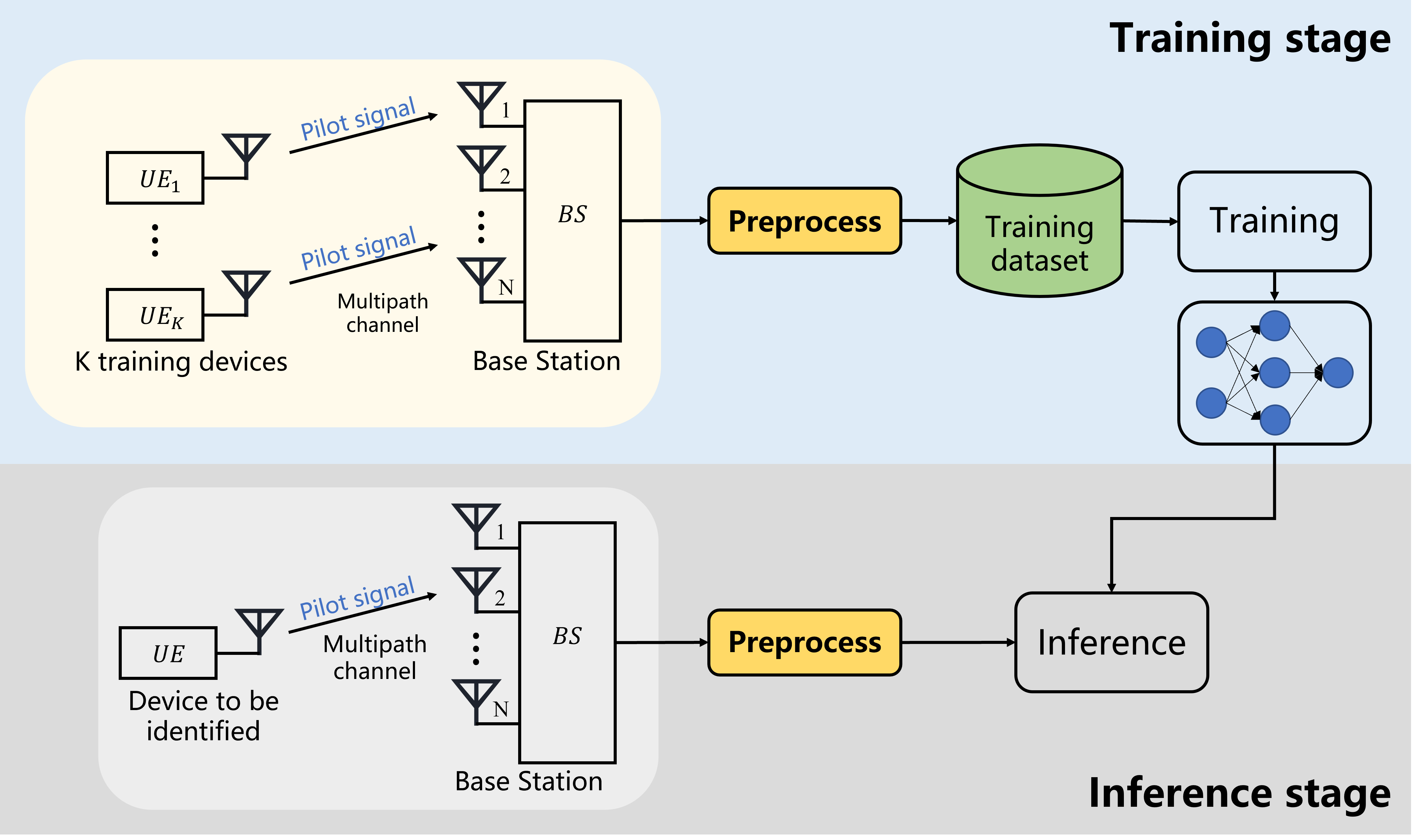}
    \caption{System overview.}
    \label{fig:system_overview}
\end{figure}

\subsection{RFF Learning Pipeline}
The training stage of the proposed method includes the extraction and preprocessing of UE’s RFF and the training of the machine learning model.

\subsubsection{Extraction and Preprocess of RFF}

The communication is conducted under the 5G framework, with \( K \) training devices transmitting pilot signals to the BS. The BS employs Algorithm~\ref{alg:ext} to extract the channel-robust RFFs of the individual devices. Upon completion of the RFF extraction, the real and imaginary parts of RFF are concatenated into a new array, which is subsequently employed for training.

\subsubsection{Model Training}

A convolutional neural network (CNN) is used to perform hierarchical RFF feature learning, which captures nonlinear high-dimensional signal relationships through stacked transformations \cite{9451544}.

\subsection{UE Identification via RFF Inference} \label{Sec:IV-B}

The UE to be identified transmits a pilot signal to BS. BS then extracts RFF using the received signals and feeds it to the pretrained model for identification.

\subsubsection{Extraction and Preprocess of RFF}

Similarly, in the inference stage, Algorithm~\ref{alg:ext} is utilized to extract the RFF of the UE for inference.

\subsubsection{Inference}

The BS puts the extracted RFF into the pre-trained CNN to determine the UE’s identity. The predicted label will be compared with the actual label to evaluate the performance of the proposed method.

\subsubsection{Identification Latency}

Fig. \ref{fig:latency} decomposes the air interface delay into terminal processing $T_{UE}$, frame alignment $T_f$, transmission time interval (TTI) $T_{TTI}$, and RFF identification $T_{RFFI}$, from which, it can be seen that the proposed scheme completes RFF extraction through only a single pilot transmission, eliminating the ms-scale transmission delay inherent in feedback-based approaches.

\begin{figure}[h]
    \centering  
    \includegraphics[width=0.8\linewidth]{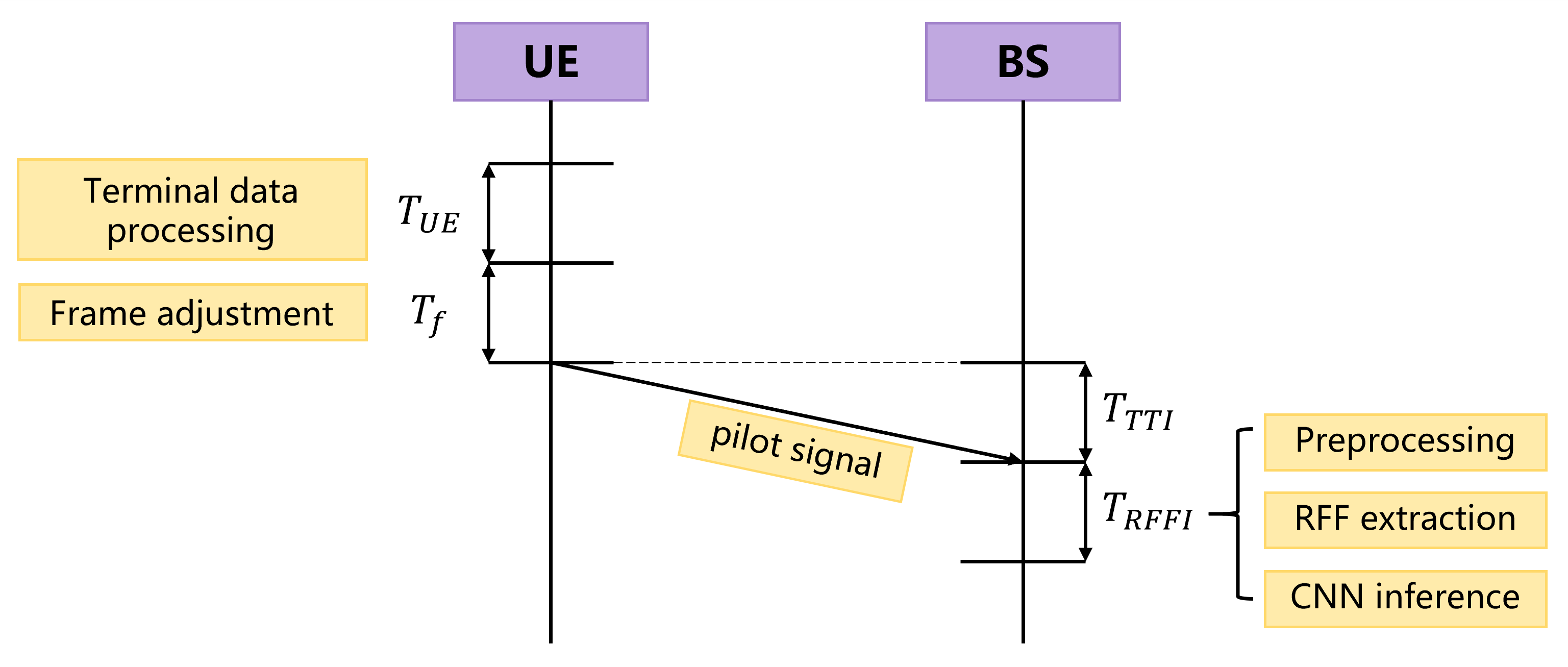}
    \caption{Air interface delay components.}
    \label{fig:latency}
\end{figure}

\section{Simulation Results}

In this section, we evaluate the performance of the proposed RFF method in terms of classification accuracy and identification latency through MATLAB simulations.

\subsection{Simulation Setup}


We conduct simulations using the MATLAB 5G NR Toolbox, modeling 30 UE devices. For each device, we collect 9000 frames across an SNR range from -10 dB to 30 dB for training and 900 frames for inference. CNN is employed to train and test the collected RFF data. 

\subsubsection{Transmitter and Receiver}



Assuming that the UE devices are single-antenna terminals and the BS is a multi-antenna base station, the BS and UE devices are configured as the RFF modeling in Section \ref{Sec:II}.

\subsubsection{Channel}


The TDL channel model specified in 3GPP TR 38.901 \cite{3GPPTR38901} is employed for simulation, configured by setting the parameters, such as \texttt{PathDelays} and \texttt{AveragePathGains}. In the simulation of wideband multipath environments, the carrier frequency $f_c$ is set to 10 GHz, the subcarrier spacing $SCS$ is set to 60 kHz, and the bandwidth $B$ is set to 10 MHz.

\subsubsection{CNN}

The CNN hierarchically learns RFF features via two cascaded blocks, each comprising a 3×3 convolutional layer (32/64 filters), batch normalization, ReLU, and 2×2 max-pooling (stride=2). Flattened features pass through a softmax-activated dense layer for 30-class prediction. Adam optimizer with 1e-3 learning rate and 32 batch size is adopted, combined with early stopping (5-epoch patience) on 20\% validation data to prevent overfitting.

\subsection{Identification Accuracy Results}


Fig.~\ref{fig:acc_multipath} demonstrates identification accuracy under varying multipath conditions. The proposed scheme achieves over 91\% accuracy when SNR exceeds 20 dB, showing robust performance against multipath effect. Notably, lower accuracy occurs in flat channels, such as 4-path channels, as RFF extraction via dual antennas requires pronounced channel differences to stabilize features; minimal differences amplify noise interference. Furthermore, accuracy does not monotonically increase with multipath components, as the 24-path channel yields lower accuracy than 20-path. 

\begin{figure}[h]
    \centering  
    \includegraphics[width=0.8\linewidth]{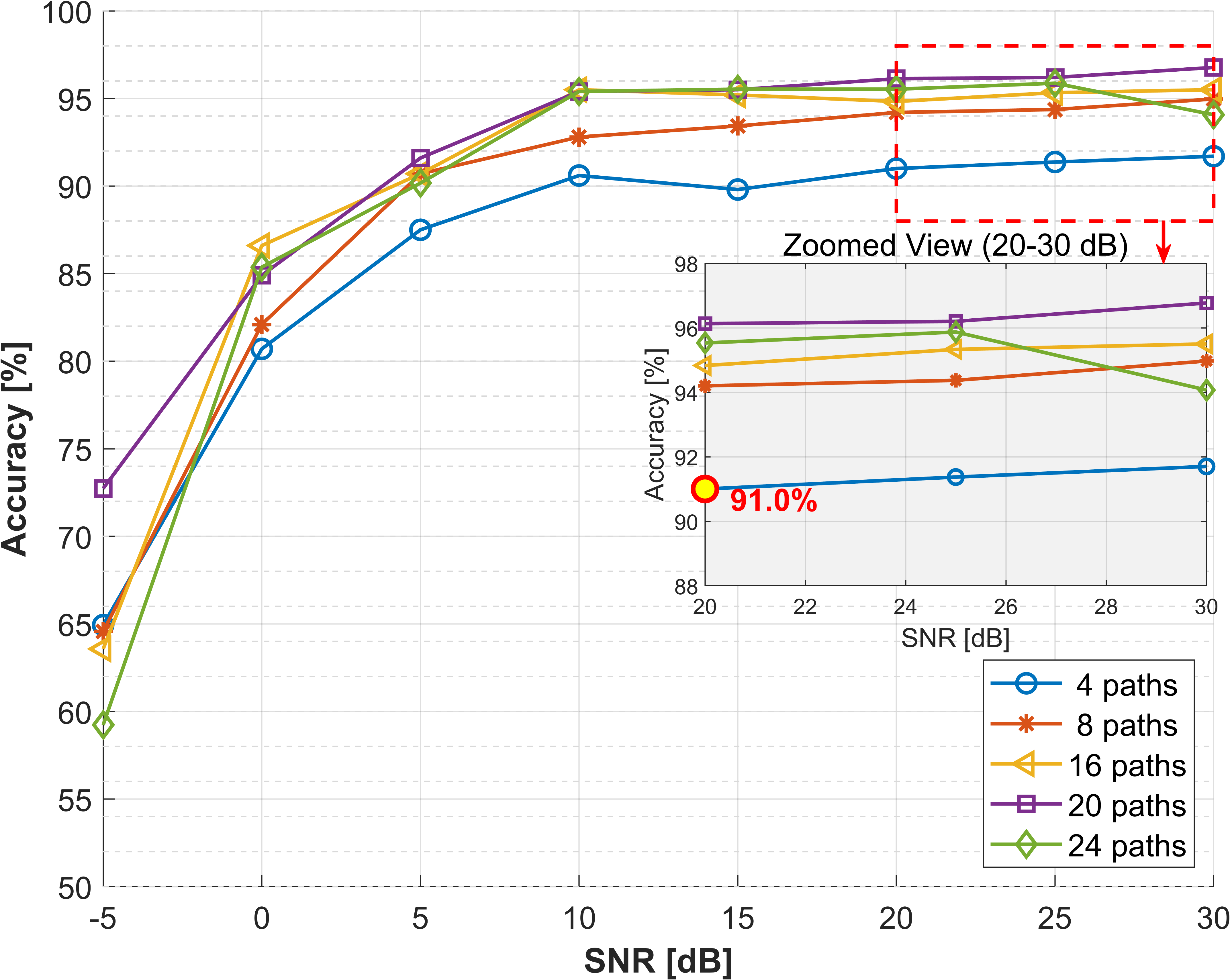}
    \caption{Identification accuracy under varying numbers of multipath components, with the frequency band divided across 16 subcarriers.}
    \label{fig:acc_multipath}
\end{figure}

Fig.~\ref{fig:acc_subcarrier} shows the identification accuracy across different numbers of subcarriers for frequency band segmentation. The scheme achieves 96.43\% accuracy at SNR over 20 dB across 4 subcarriers, as narrower sub-bands yield smaller differences between the estimated and true mean value of CFRs within sub-bands, which enhance identification accuracy.

\begin{figure}[h]
    \centering  
    \includegraphics[width=0.8\linewidth]{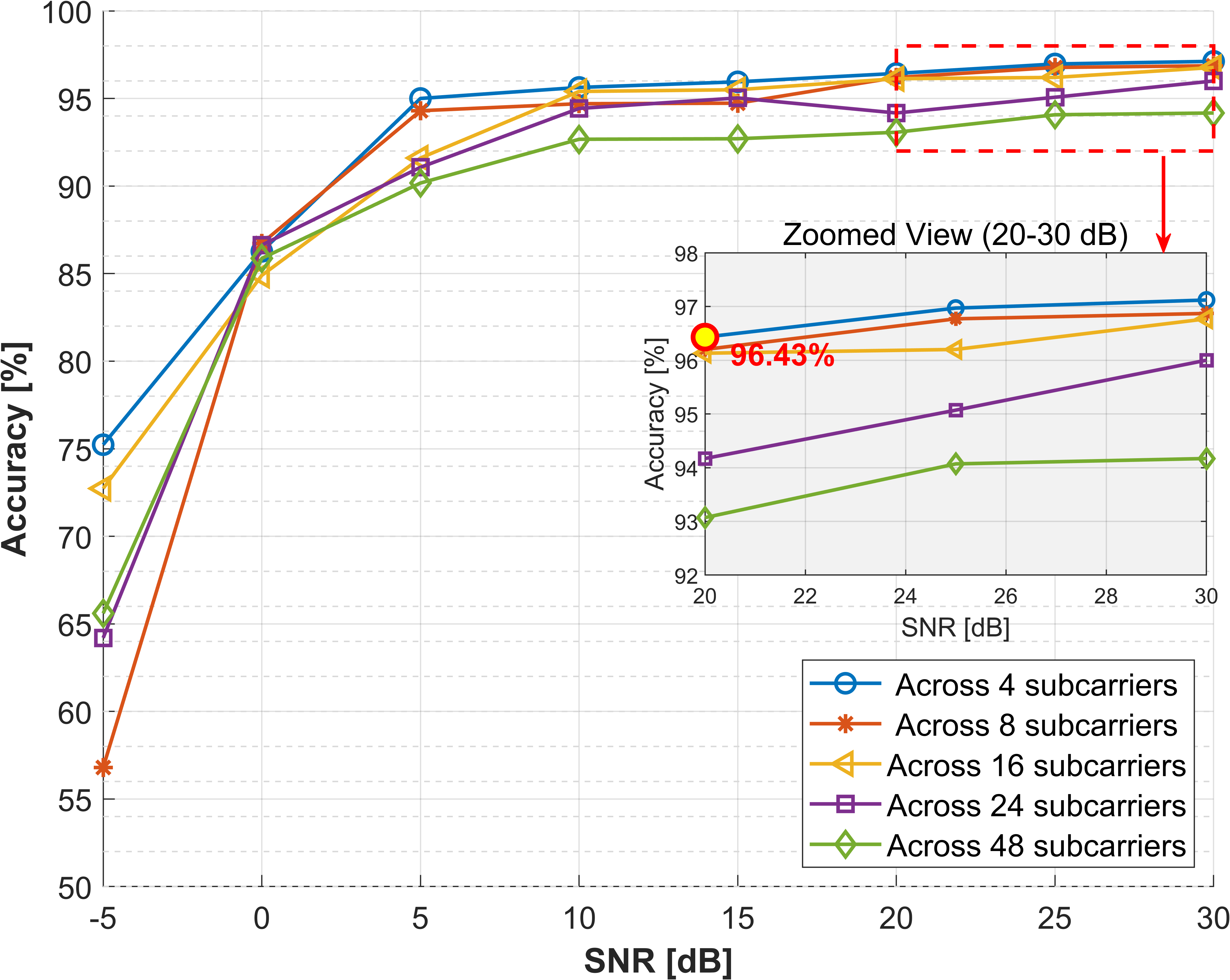}
    \caption{Identification accuracy across different numbers of subcarriers under the 20-path channel.}
    \label{fig:acc_subcarrier}
\end{figure}

As shown in Fig.~\ref{fig:acc_device}, the identification accuracy of the proposed scheme can achieve 98.2\%, when the number of UEs is less than 20 and the SNR exceeds 15 dB. As the number of UEs increases, the identification accuracy declines.
When the number of UEs is 30 and the SNR is greater than 10 dB, the identification accuracy can reach above 95.4\%, thereby demonstrating the satisfactory classification performance for 30 or fewer UEs.

\begin{figure}[h]
    \centering  
    \includegraphics[width=0.8\linewidth]{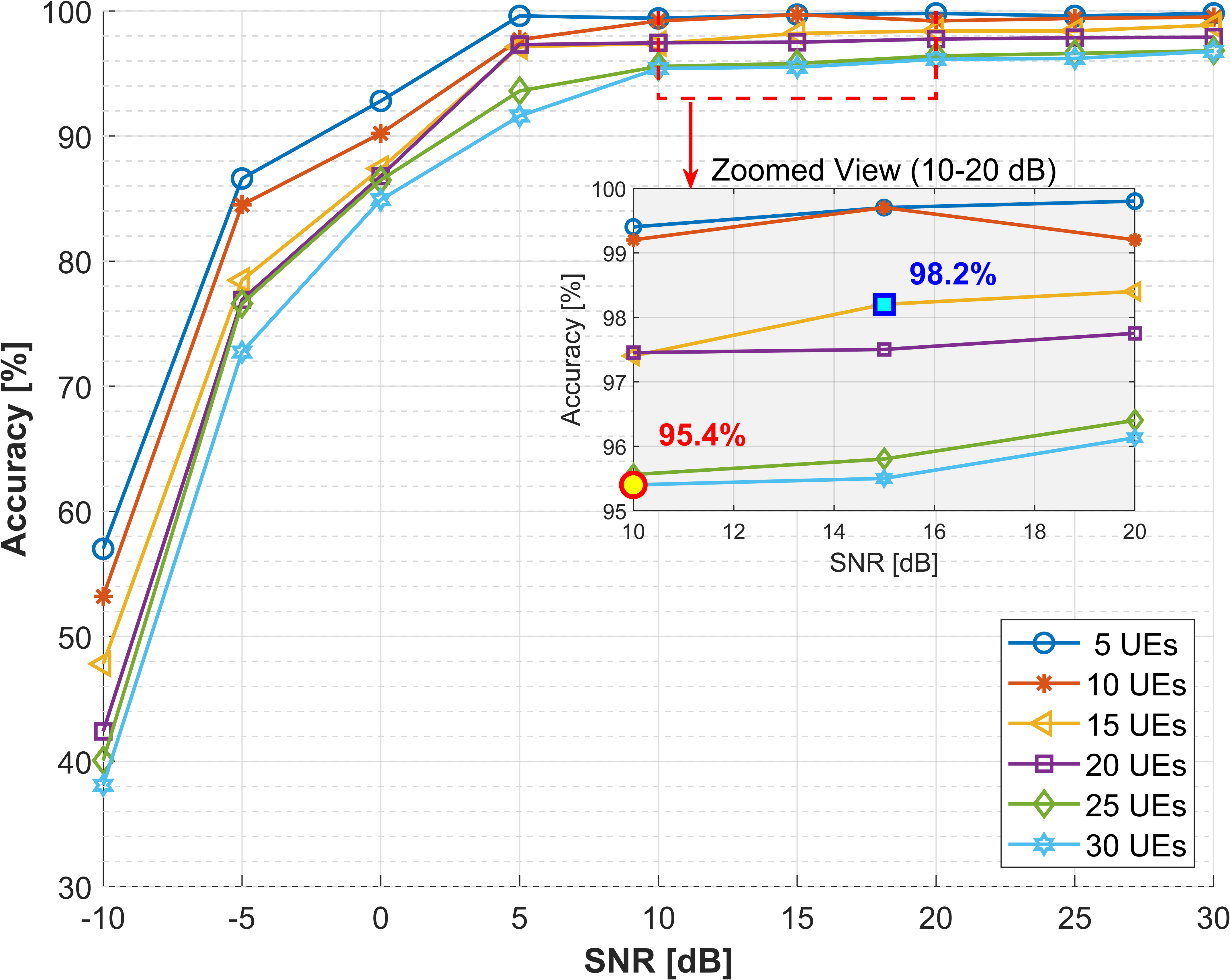}
    \caption{Identification accuracy of the different number of UEs under the 20-path channel, with the frequency band divided across 16 subcarriers.}
    \label{fig:acc_device}
\end{figure}

Fig.~\ref{fig:acc_compare} compares identification accuracy across methods, including the logarithmic-domain DoLoS scheme \cite{9979789}, Benchmark1 and Benchmark2. Benchmark1 directly processes raw IQ data through CNN, while Benchmark2 adds preprocessing without frequency band division. At 25 dB SNR, the proposed method outperforms all baselines with 96.2\% accuracy, compared to 61.83\% of Benchmark1, 62.4\% of Benchmark2, and 83.67\% of DoLoS. This demonstrates Benchmark1's sensitivity to multipath distortion and Benchmark2's limitation from CFR fluctuation. Although DoLoS shows improved robustness, its logarithmic-domain operation degrades accuracy for partial RFF feature loss.

\begin{figure}[h]
    \centering  
    \includegraphics[width=0.8\linewidth]{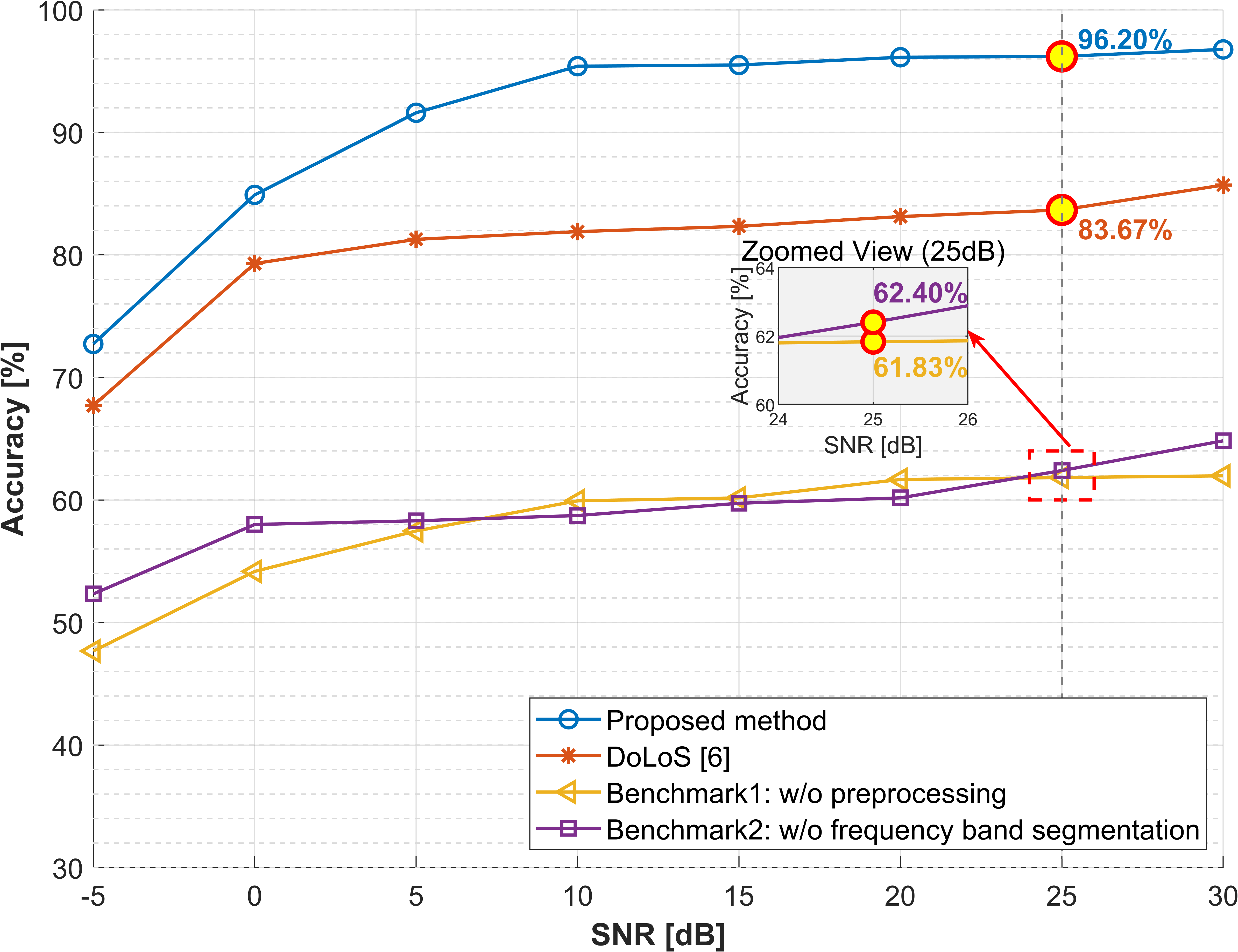}
    \caption{Identification accuracy of different methods under 20-path channel.}
    \label{fig:acc_compare}
\end{figure}

\subsection{Identification Latency Estimation }

To evaluate the low-latency characteristic, we measure the device identification time. However, system, interpreter, and resource contention skew the MATLAB measurement results. Therefore, we estimate the theoretical latency based on the Roofline model \cite{roofline}. 
Following the air interface delay components presented in Sec \ref{Sec:IV-B}, 
the simulation program is partitioned into code blocks. 
Each block is classified as either compute-bound or memory-bound using the Roofline model, and the corresponding execution times are estimated.


Based on the hardware parameters of the test computer, the computed peak performance is $Peak\,FLOPS = 403.2\,\mathrm{GFLOPS}$ and $Memory\,Bandwidth = 17.06\,\mathrm{GB/s}$.
Table \ref{tab:simu} lists the total arithmetic operations and memory access required by the UE and the BS for their respective data-processing stages. From these values, we estimate  
$T_{UE} = \SI{34.64}{\micro\second}$ and  
$T_{RFFI} = \SI{81.04}{\micro\second}$.  
Furthermore, with the simulation configured at $SCS = 60\,\mathrm{kHz}$, we obtain $T_{TTI} = 0.25\,\mathrm{ms}$ and $T_f = 0.125\,\mathrm{ms}$.  
Collectively, the overall theoretical latency is  
$T_{air} = T_{UE} + T_f + T_{TTI} + T_{RFFI} \approx 0.491\,\mathrm{ms} < 1\,\mathrm{ms}$,  
thereby satisfying the URLLC latency requirement specified in ITU-R IMT-2020 \cite{ITU2017b}.


\begin{table}[htbp]
  \centering
  \caption{The computational workload and the memory access of the simulation program}\label{tab:simu}
  \begin{tabular}{lcc}
    \toprule
    \textbf{Program block} & \textbf{Total FLOPS (FLOPS)} & \textbf{Total Bytes (Bytes)} \\
    \midrule
    UE side & $7.34 \times 10^5$ & $5.91 \times 10^5$ \\
    BS side & $1.37 \times 10^6$ & $1.39 \times 10^6$ \\
    \bottomrule
  \end{tabular}
\end{table}

\section{Conclusion}


This paper proposes a novel LLDR-based RFF identification scheme to tackle multipath distortion without additional delay. By utilizing co-temporal CFR measurements from multiple antennas and incorporating sub-band processing, the method effectively preserves device-specific RFF under minimal channel variation. Simulation results confirm that the proposed method enables effective device identification under multipath channel conditions while fully complying with ultra-low latency requirements.

\begin{appendices}

\end{appendices}

\footnotesize
\bibliographystyle{IEEEtran}
\bibliography{references}

\vfill

\end{document}